\documentclass[graybox]{svmult}

\usepackage{helvet}         
\usepackage{courier}        
\usepackage{type1cm}        
\usepackage{makeidx}         
\usepackage{graphicx}        
\usepackage{multicol}        
\usepackage[bottom]{footmisc}

\usepackage{amsmath,amssymb}

\newcommand{\Natural}{\mathbb{N}}
\newcommand{\FF}{\mathcal{F}}
\newcommand{\GG}{Q}
\newcommand{\HH}{P}
\newcommand{\dx}{\mathrm d^4x}
\newcommand{\nmu}{\nabla_{\mu}}
\newcommand{\nnu}{\nabla_{\nu}}
\DeclareMathOperator{\Hess}{Hess}

\begin{document}

\title*{Variations of Infinite Derivative Modified  Gravity}
\author{Ivan Dimitrijevic, Branko Dragovich, Zoran Rakic and Jelena Stankovic}
\institute{Ivan Dimitrijevic \at Faculty of Mathematics, University of Belgrade,  Belgrade, Serbia, \email{ivand@matf.bg.ac.rs}
\and Branko Dragovich \at Institute of Physics, University of Belgrade,  Belgrade, Serbia,\\ Mathematical Institute, Serbian Academy of Sciences and Arts, Belgrade, Serbia \email{dragovich@ipb.ac.rs}
\and Zoran Rakic \at Faculty of Mathematics, University of Belgrade,  Belgrade, Serbia, \email{zrakic@matf.bg.ac.rs}
\and Jelena Stankovic \at Teacher Education Faculty, University of Belgrade, Belgrade, Serbia \email{jelenagg@gmail.com}}
\maketitle

\abstract{We consider nonlocal modified Einstein gravity without matter, where nonlocal term has the
form  $\HH(R) \FF(\Box) \GG(R)$. For this model, in this paper we give the derivation of the equations of motion
in detail. This is not  an easy  task and  presented derivation should be useful to a researcher who wants  to investigate nonlocal gravity.
Also, we present the second variation of the related Einstein-Hilbert modified action and basics of gravity perturbations.}

\section{Introduction}
General relativity \cite{wald}, which is Einstein theory of gravity, is dominant theory of gravitational phenomena for more than last hundred years.
It is one of the most attractive and phenomenologically successful physical theories. General relativity is perfectly confirmed in the Solar system. Among many important predictions are
gravitational red shift, gravitational lensing, gravitational waves and black holes.

Although very successful, Einstein  gravity is not a complete and final theory of gravitational phenomena. According to its cosmological solutions, the universe contains initial singularity. This singularity is a serious and still unsolved problem, which requires an adequate Einstein gravity modification. General relativity  predicts that the universe contains about $68\%$ of {\it dark energy}, $27\%$ of {\it dark matter} and only about $5\%$ of {\it visible matter}. However, dark energy and dark matter are not yet experimentally confirmed. Also, Einstein gravity has not been verified at very large cosmic scales. Hence, some cosmological predictions, including energy/matter content of the universe, based on Einstein gravity should be taken with caution. One has also to mention problem of quantization of general relativity. Hence, it follows that Einstein gravity has some theoretical and phenomenological problems in ultraviolet and infrared regions.

Unfortunately, a new theoretical principle which would tell us  which is right  direction to modify gravity  is not yet uncovered.
Hence, there are  many  modifications of Einstein gravity,
which are motivated by problems in quantum gravity, string theory, astrophysics and cosmology (for a  review, see \cite{clifton,nojiri,novello,faraoni,nojiri1}).

One of recent and very promising directions of research is {\it nonlocal modified gravity} with its applications to cosmology (as a review, see
 \cite{nojiri,nojiri1,woodard,maggiore,dimitrijevic6,dragovich0}). Potentially there is a huge number of possibilities to construct a nonlocal gravity model  by  replacement of the scalar
curvature $R$ in the Einstein-Hilbert action by a scalar
function $F (R, \Box, \Box^{-1}, R_{\mu\nu}R^{\mu\nu}, R_{\mu\nu\alpha\beta}R^{\mu\nu\alpha\beta}, \ldots)$, where $\Box = \nabla_{\mu}\nabla^{\mu}$
is d'Alembert operator and $\nabla_{\mu}$ denotes the covariant derivative. In this paper, nonlocality means that gravitational Lagrangian
contains an infinite number of space-time derivatives, i.e.
derivatives up to an infinite order in the form of d'Alembert
operator $\Box,$ which is argument of an analytic function.  Note that higher derivative gravity theories improve problems with quantization of general relativity, see, e.g. \cite{stelle,koshelev-1,modesto1,modesto2}.

In the sequel we consider nonlocal modification of gravity where Einstein-Hilbert action contains an additional nonlocal term of the form $\HH(R) \FF(\eta\Box)\GG(R),$ where $\eta$ is a parameter and $\eta\Box$ is dimensionless.
In fact we consider a class of nonlocal gravity models without matter
given by the  action
\begin{align}\label{action}
  S &= \frac{1}{16\pi G} \int_M \left(R-2\Lambda + \HH(R) \FF(\eta \Box) \GG(R) \right)\; \sqrt{-g} \; \dx,
\end{align}
where $M$ is a pseudo-Riemannian manifold of signature $(1,3)$ with metric $(g_{\mu\nu})$, $\mathcal{F}(\eta\Box)= \displaystyle \sum_{n =0}^{\infty}
f_{n} (\eta\Box)^{n}$, $\HH$ and $\GG$ are differentiable functions of
the scalar curvature $R$ and $\Lambda$ is cosmological constant. Inspiration for an analytic form of the function $\mathcal{F}(\eta\Box)$ comes from string theory, in particular from $p$-adic string theory, which is a part of $p$-adic mathematical physics, for a recent review see \cite{dragovich1}.  The corresponding
Einstein equations of motion have complex structure. In this paper we will
present their derivation, because it is not an easy task and it should be useful to a reader interested in this subject. It is  also useful  to see \cite{biswas4}. In order to obtain equations of motion for
$g_{\mu\nu}$ we have to find the variation of the action
\eqref{action} with respect to metric $g^{\mu\nu}$. In addition we also find the second variation of the action \eqref{action} and consider some cosmological perturbations. For simplicity, in the sequel we shall take $\eta = 1.$

Before to proceed with derivation of equations of motion for the above model \eqref{action}, it is worth to mention some other nonlocal models with inverse d'Alembert operator, i. e. with
$ \Box^{-n},$ which are proposed to explain the late time cosmic acceleration without dark energy. Such models have the form
\begin{equation} \label{action2}
S = \frac{1}{16\pi G} \int \sqrt{-g} \left(  R + L_{NL}  \right) \, d^4x ,
\end{equation}
where two typical examples are: $L_{NL} = R\, f(\Box^{-1}R)$ (see a review \cite{nojiri,woodard} and references therein), and  $ L_{NL} = -\frac{1}{6} m^2 R \Box^{-2} R $ (see a review \cite{maggiore} and references therein).

Nonlocal models with $\mathcal{F}(\Box)= \displaystyle \sum_{n =0}^{\infty}
f_{n} \Box^{n}$ are mainly considered to improve general relativity in its ultraviolet region, unlike models with $ \Box^{-1}$ and $ \Box^{-2}$ which intend to modify  gravity in its infrared sector. It may happen that there will be more than one modification of general relativity, which are valid at the different scales. Namely, any physical theory has a domain of validity, which depends on some conditions, including spatial scale and  complexity of the system. It is natural that validity of general relativity is also restricted. At very short and very large cosmic distances may act different gravity theories.

  Section 2 contains variation of curvature tensors from pseudo-Riemannian geometry, what is necessary for derivation of equations of motion for $g_{\mu\nu}$ in Sect. 3. Second variation of  gravity modified action \eqref{action} is presented in Sect. 4. Basics of cosmic perturbations are subject of Sect. 5. Sect. 6 contains some concluding remarks.

\section{Variation of curvature tensors}

Let us start with a technical lemma:
\begin{lemma} \label{lem:deltag1}
  The following relations hold
  \begin{align}
    \delta g &= g g^{\mu\nu} \delta g_{\mu\nu} = - g g_{\mu\nu} \delta g^{\mu\nu},\label{lem:deltag.1} \\
    \delta \sqrt{-g} &= - \frac 12 g_{\mu\nu} \sqrt{-g} \delta g^{\mu\nu}, \label{lem:deltag.2}\\
    \delta \Gamma_{\mu\nu}^\lambda &= -\frac 12 \left( g_{\nu\alpha} \nabla_\mu \delta g^{\lambda \alpha} + g_{\mu\alpha} \nabla_\nu \delta g^{\lambda \alpha} - g_{\mu\alpha} g_{\nu\beta} \nabla^\lambda \delta g^{\alpha\beta}\right),\label{lem:deltag.3}
  \end{align}

  where  $g$ is  the determinant of the metric tensor. \\
\end{lemma}

\begin{proof}
\vspace{-5mm}
Determinant $g$  can be written as
\begin{equation}
g_{\mu\nu}G^{(\alpha,\nu)}= g \delta_{\mu}^{\alpha},
\end{equation}
where $ G^{(\mu,\nu)}$ is the algebraic cofactor of the element $g_{\mu\nu}$.\\
Thus,
\begin{equation}
g^{\mu\nu}= \frac{G^{(\mu,\nu)}}{g}.
\end{equation}
Since  $G^{\mu,\nu}$ is independent of $g_{\mu\nu}$ we obtain the first part of the equation \eqref{lem:deltag.1}
\begin{equation}
\delta g = g g^{\mu\nu}\delta g_{\mu\nu}.
\end{equation}
Moreover from  $ g_{\mu\nu}g^{\mu\nu}= 4 $ and Leibniz rule we obtain $g^{\mu\nu}\delta
g_{\mu\nu}= - g_{\mu\nu}\delta g^{\mu\nu}$, which completes the proof of \eqref{lem:deltag.1}. \\
To prove the equation \eqref{lem:deltag.2} we proceed in the following way
\begin{equation}
\delta \sqrt{-g}= - \frac{1}{2 \sqrt{-g}}\delta g = - \frac 12 g_{\mu\nu} \sqrt{-g} \delta g^{\mu\nu}.
\end{equation}
The third equation is proved by
\begin{align}
  \nabla_\lambda g_{\mu\nu} &= \partial_\lambda g_{\mu\nu} - \Gamma_{\lambda\mu}^\kappa g_{\kappa\nu} - \Gamma_{\lambda\nu}^\kappa g_{\mu\kappa}, \\
  \delta \Gamma_{\mu\nu}^{\lambda}&= \frac 12 \delta g^{\lambda\kappa}\left(\partial_\mu g_{\kappa\nu} + \partial_\nu g_{\mu\kappa} - \partial_\kappa g_{\mu\nu} \right)+ \frac 12 g^{\lambda\kappa}\left(\partial_\mu \delta g_{\kappa\nu} + \partial_\nu \delta g_{\mu\kappa} - \partial_\kappa \delta g_{\mu\nu} \right).
\\
\delta \Gamma_{\mu\nu}^{\lambda}&= \frac 12 \delta g^{\lambda\kappa}\left(\partial_\mu g_{\kappa\nu} + \partial_\nu g_{\mu\kappa} - \partial_\kappa g_{\mu\nu} \right)+ \frac 12 g^{\lambda\kappa}\left(\partial_\mu \delta g_{\kappa\nu} + \partial_\nu \delta g_{\mu\kappa} - \partial_\kappa \delta g_{\mu\nu} \right) \\
& = \frac 12 \delta g^{\lambda\kappa}\left(\partial_\mu g_{\kappa\nu} + \partial_\nu g_{\mu\kappa} - \partial_\kappa g_{\mu\nu} \right)
+ \frac 12 g^{\lambda\kappa} \left(\nabla_\mu\delta g_{\kappa\nu}+ \nabla_{\nu}\delta g_{\mu\kappa} - \nabla_{\kappa}\delta g_{\mu\nu}  \right) \nonumber \\
& + g^{\lambda\kappa} \Gamma_{\mu\nu}^{\alpha}\delta g_{\kappa \alpha} \\
& = \frac 12 g^{\lambda\kappa} \left(\nabla_\mu\delta g_{\kappa\nu}+ \nabla_{\nu}\delta g_{\mu\kappa} - \nabla_{\kappa}\delta g_{\mu\nu}  \right).
\end{align}
In the last step we used $\delta g_{\mu\nu}= - g_{\mu \alpha}g_{\nu \beta}\delta g^{\alpha\beta}$. Using the same equation in every term of the last equation we obtain \eqref{lem:deltag.3}:
  \begin{align}
\delta \Gamma_{\mu\nu}^{\lambda}&= \frac 12 g^{\lambda\kappa} \left(\nabla_\mu(-g_{\kappa\alpha} g_{\beta\nu} \delta g^{\alpha\beta})+ \nabla_{\nu}(-g_{\mu\alpha}g_{\kappa\beta}\delta g^{\alpha\beta}) - \nabla_{\kappa}(-g_{\mu\alpha}g_{\nu\beta}\delta g^{\alpha\beta})  \right) \\
&= -\frac 12 \left(\delta_\alpha^\lambda g_{\beta\nu}\nabla_\mu\delta g^{\alpha\beta}+ \delta_{\beta}^\lambda g_{\mu\alpha}\nabla_{\nu}\delta g^{\alpha\beta} - g_{\mu\alpha}g_{\nu\beta}\nabla^{\lambda}\delta g^{\alpha\beta}  \right) \\
&= -\frac 12 \left( g_{\nu\alpha} \nabla_\mu \delta g^{\lambda \alpha} + g_{\mu\alpha} \nabla_\nu \delta g^{\lambda \alpha} - g_{\mu\alpha} g_{\nu\beta} \nabla^\lambda \delta g^{\alpha\beta}\right).
\end{align}

Note, that $\delta \Gamma_{\mu\lambda}^{\lambda} = -\frac 12 g_{\lambda\alpha} \nabla_\mu \delta g^{\lambda \alpha}$.
\end{proof}

\begin{lemma} \label{lem:deltag}
  The variation of Riemman tensor, Ricci tensor and scalar curvature satisfy the following relations
  \begin{align}
    \delta R_{\mu\beta\nu}^\alpha &= \nabla_\beta \delta \Gamma_{\mu\nu}^\alpha - \nabla_\nu \delta \Gamma_{\mu \beta}^\alpha,\label{lem:deltag:4a} \\
    \delta R_{\mu\nu} &= \nabla_\lambda \delta \Gamma_{\mu\nu}^\lambda - \nabla_\nu \delta \Gamma_{\mu\lambda}^\lambda,\label{lem:deltag:4b} \\
    \delta R &= R_{\mu\nu} \delta g^{\mu\nu} - K_{\mu\nu} \delta g^{\mu\nu}, \label{lem:deltag.4}\\
    \delta \nabla_\mu \nabla_\nu \psi &= \nabla_\mu\nabla_\nu \delta \psi - \nabla_\lambda \psi \delta \Gamma_{\mu\nu}^\lambda, \label{lem:deltag.5}
  \end{align}

  where $K_{\mu\nu} = \nabla_\mu \nabla_\nu - g_{\mu\nu}\Box$. \\
\end{lemma}
\begin{proof}
The variation of Riemann tensor is obtained as follows
\begin{equation}
  \begin{aligned}
    \delta R_{\mu\beta\nu}^\alpha &= \delta \left(\partial_\beta \Gamma_{\mu\nu}^\alpha - \partial_\nu \Gamma_{\mu\beta}^\alpha  + \Gamma_{\mu\nu}^\lambda \Gamma_{\beta\lambda}^\alpha - \Gamma_{\mu\beta}^\lambda \Gamma_{\nu\lambda}^\alpha\right)\\
    &= \partial_\beta \delta\Gamma_{\mu\nu}^\alpha - \partial_\nu \delta \Gamma_{\mu\beta}^\alpha  + \delta \Gamma_{\mu\nu}^\lambda \Gamma_{\beta\lambda}^\alpha + \Gamma_{\mu\nu}^\lambda \delta \Gamma_{\beta\lambda}^\alpha -  \Gamma_{\mu\beta}^\lambda \delta \Gamma_{\nu\lambda}^\alpha \\
    &- \delta \Gamma_{\mu\beta}^\lambda \Gamma_{\nu\lambda}^\alpha - \Gamma_{\beta\nu}^\lambda \delta \Gamma_{\mu\lambda}^\alpha + \Gamma_{\beta\nu}^\lambda \delta \Gamma_{\mu\lambda}^\alpha .\\
  \end{aligned}
\end{equation}

In the last equation first, third and fifth term combined give $\nabla_\beta \delta\Gamma_{\mu\nu}^\alpha$, and second, fourth and sixth term give $-\nabla_\nu \delta\Gamma_{\mu\beta}^\alpha$, so we proved \eqref{lem:deltag:4a}. Equation \eqref{lem:deltag:4b} is obtained from the previous by contracting indices $\alpha$ and $\beta$.

To prove the equation \eqref{lem:deltag.4} we begin  with $R =  g^{\mu\nu} R_{\mu\nu}$. Applying the operator $\delta$ to both sides yields
\begin{equation}\begin{aligned}
  \delta R &= R_{\mu\nu} \delta g^{\mu\nu} + g^{\mu\nu} \delta R_{\mu\nu} \\
           &= R_{\mu\nu} \delta g^{\mu\nu} + g^{\mu\nu} \left( \nabla_\lambda \delta\Gamma_{\mu\nu}^\lambda - \nabla_\nu\delta\Gamma_{\mu\lambda}^\lambda \right) \\
           &= R_{\mu\nu} \delta g^{\mu\nu} -\frac 12 g^{\mu\nu} \left(2 g_{\nu\alpha} \nabla_\lambda \nabla_\mu \delta g^{\lambda \alpha} - g_{\mu\alpha} g_{\nu\beta} \Box \delta g^{\alpha\beta} - g_{\lambda\alpha} \nabla_\nu \nabla_\mu \delta g^{\lambda \alpha}\right) \\
           &= R_{\mu\nu} \delta g^{\mu\nu} -\frac 12 \left(2 \delta_{\alpha}^\mu \nabla_\lambda \nabla_\mu \delta g^{\lambda \alpha} - \delta_{\alpha}^\nu g_{\nu\beta} \Box \delta g^{\alpha\beta} - g_{\lambda\alpha} \Box \delta g^{\lambda \alpha}\right) \\
           &= R_{\mu\nu} \delta g^{\mu\nu} -\frac 12 \left(2 \nabla_\mu \nabla_\nu \delta g^{\mu\nu} -  2g_{\mu\nu} \Box \delta g^{\mu\nu} \right) \\
           &= R_{\mu\nu} \delta g^{\mu\nu} -K_{\mu\nu} \delta g^{\mu\nu}.
\end{aligned} \end{equation}

The last equation \eqref{lem:deltag.5} is proved in the following way
  \begin{align}
        \delta \nabla_\mu \nabla_\nu \psi &= \delta \left(\partial_\mu \nabla_\nu \psi - \Gamma_{\mu\nu}^\lambda\nabla_\lambda \psi\right) \nonumber \\
        &= \delta \left(\partial^2_{\mu\nu} \psi - \Gamma_{\mu\nu}^\lambda\partial_\lambda \psi\right) \nonumber\\
        &= \partial^2_{\mu\nu} \delta \psi - \Gamma_{\mu\nu}^\lambda\partial_\lambda \delta \psi - \partial_\lambda \psi \delta \Gamma_{\mu\nu}^\lambda \\
        &= \nabla_\mu\nabla_\nu \delta \psi -  \nabla_\lambda \psi \; \delta \Gamma_{\mu\nu}^\lambda. \nonumber 
  \end{align}
\end{proof}
\begin{lemma} \label{lem:hkdeltag}
  Every scalar function $\HH(R)$ satisfies
    \begin{align}
        \int_M \HH g_{\mu\nu} (\Box \delta g^{\mu\nu}) \sqrt{-g} \; \dx &= \int_M g_{\mu\nu} (\Box \HH)  \delta g^{\mu\nu} \sqrt{-g}\; \dx, \label{lem:hkdeltag.1}\\
        \int_M \HH \nabla_\mu \nabla_\nu \delta g^{\mu\nu} \sqrt{-g} \; \dx & = \int_M \nabla_\mu \nabla_\nu \HH \; \delta g^{\mu\nu} \sqrt{-g} \; \dx, \label{lem:hkdeltag.2}\\
        \int_M \HH K_{\mu\nu} \delta g^{\mu\nu} \sqrt{-g} \; \dx & = \int_M K_{\mu\nu} \HH \; \delta g^{\mu\nu} \sqrt{-g} \; \dx. \label{lem:hkdeltag.3}
    \end{align}

\end{lemma}

\begin{proof}
Equation \eqref{lem:hkdeltag.1} is proved by application of Stokes' theorem:
\begin{equation}\begin{aligned}
\int_M \HH g_{\mu\nu} \Box \delta g^{\mu\nu} \sqrt{-g} \; \dx &= \int_M \HH g_{\mu\nu} \nabla_{\alpha}\nabla^{\alpha} \delta g^{\mu\nu} \sqrt{-g} \; \dx \\
&=-\int_M \nabla_{\alpha}(\HH g_{\mu\nu})\nabla^{\alpha}\delta g^{\mu\nu} \sqrt{-g} \; \dx  \\
&= \int_M g_{\mu\nu}\nabla^{\alpha}\nabla_{\alpha}\HH \; \delta g^{\mu\nu} \sqrt{-g}\; \dx \\
&= \int_M g_{\mu\nu} \Box \HH \; \delta g^{\mu\nu}\sqrt{-g} \; \dx.
\end{aligned} \end{equation}

Let $N^\mu = \HH \nabla_\nu \delta g^{\mu\nu} - \nabla_\nu \HH \delta g^{\mu\nu}$, then $ \nmu N^\mu$ can be written as
\begin{equation}\begin{aligned}
\nmu N^{\mu} & = \nmu (\HH \nnu \delta g^{\mu\nu} - \nnu \HH \delta g^{\mu\nu}) \\
 & = \nmu \HH \nnu \delta g^{\mu\nu} + \HH \nmu \nnu \delta g^{\mu\nu} - \nmu \nnu \HH \; \delta g^{\mu\nu} - \nnu \HH \nmu \delta g^{\mu\nu} \\
& = \HH \nmu \nnu \delta g^{\mu\nu} - \nmu \nnu \HH \; \delta
g^{\mu\nu}.
\end{aligned}\end{equation}

Integration over $M$ yields $\int_M \nmu N^{\mu} \sqrt{-g} \; \dx = \int_{\partial M} N^{\mu}n_{\mu} \mathrm d\partial M $, where $n_\mu$ is the unit normal to a hypersurface $\partial M$. As the restriction  $N^\mu|_{\partial M}$  vanish, the last integral vanish as well, which proves \eqref{lem:hkdeltag.2}.

Equation \eqref{lem:hkdeltag.3} is a direct consequence of \eqref{lem:hkdeltag.1} and \eqref{lem:hkdeltag.2}.
\end{proof}
\begin{lemma}\label{lem:j_n}
Let $\HH(R)$ and $\GG(R)$ be scalar functions. Then for all $n \in \Natural$
\begin{align}
\int_M \HH \delta \Box^n \GG \sqrt{-g} \; \dx & =  \frac 12 \sum_{l=0}^{n-1}\int_M S_{\mu\nu}(\Box^l \HH,\Box^{n-1-l} \GG) \delta g^{\mu\nu} \; \sqrt{-g} \; \dx \nonumber \\
&+ \int_M \Box^n \HH \; \delta \GG \sqrt{-g} \; \dx.
\end{align}
\end{lemma}
\begin{proof}
The definition of the $\Box$ operator implies
\begin{equation}\begin{aligned}
I=&\int_M \HH \delta \Box^n \GG \sqrt{-g} \; \dx = \int_M \HH \delta (g^{\mu\nu}\nabla_\mu\nabla_\nu  \Box^{n-1}\GG) \sqrt{-g} \; \dx \\
=& \int_M \HH \left(\nabla_\mu\nabla_\nu\; \Box^{n-1}\GG \delta g^{\mu\nu} + g^{\mu\nu} \delta \nabla_\mu\nabla_\nu  \Box^{n-1}\GG\right) \sqrt{-g} \; \dx \\
=& \int_M \HH \left(\nabla_\mu\nabla_\nu\; \Box^{n-1}\GG \delta g^{\mu\nu} + \Box  \delta \Box^{n-1}\GG - \nabla_\lambda \Box^{n-1}\GG g^{\mu\nu} \delta \Gamma_{\mu\nu}^\lambda \right) \sqrt{-g} \; \dx \label{eq:1}.
\end{aligned} \end{equation}

On the other hand, Lemma 1 
yields
\begin{equation}\begin{aligned}
  g^{\mu\nu} \delta \Gamma_{\mu\nu}^\lambda &= -\frac 12 g^{\mu\nu}\left( g_{\nu\alpha} \nabla_\mu \delta g^{\lambda \alpha} + g_{\mu\alpha} \nabla_\nu \delta g^{\lambda \alpha} - g_{\mu\alpha} g_{\nu\beta} \nabla^\lambda \delta g^{\alpha\beta}\right) \\
  &= -\frac 12 \left( \delta_\alpha^\mu \nabla_\mu \delta g^{\lambda \alpha} + \delta_\alpha^\nu \nabla_\nu \delta g^{\lambda \alpha} - \delta_\alpha^\nu g_{\nu\beta} \nabla^\lambda \delta g^{\alpha\beta}\right) \\
  &= - \frac 12 (2\nabla_\mu \delta g^{\lambda \mu} - g_{\mu\nu} \nabla^\lambda \delta g^{\mu\nu}).
\end{aligned}\end{equation}

Moreover, from the equation \eqref{eq:1} and Stokes' theorem we get
\begin{equation}\begin{aligned}
I=&\int_M \HH \Big(\nabla_\mu\nabla_\nu \Box^{n-1}\GG \delta g^{\mu\nu} + \Box  \delta \Box^{n-1}\GG \\
+&\frac 12 \nabla_\lambda \Box^{n-1}\GG (
2\nabla_\mu \delta g^{\lambda \mu} - g_{\mu\nu} \nabla^\lambda \delta g^{\mu\nu}) \Big) \sqrt{-g} \; \dx \\
=& \int_M \HH \nabla_\mu\nabla_\nu \Box^{n-1}\GG \delta g^{\mu\nu} \sqrt{-g} \; \dx + \int_M \HH\;\Box \delta \Box^{n-1}\GG \sqrt{-g} \; \dx \\
-& \int_M \nabla_\mu(\HH \nabla_\lambda \Box^{n-1}\GG )\delta g^{\lambda\mu} \sqrt{-g} \; \dx \\
-& \frac 12\int_M g_{\mu\nu}\nabla^{\lambda}(\HH \nabla_\lambda \Box^{n-1}\GG) \delta g^{\mu\nu} \sqrt{-g} \; \dx \\
=&  \int_M \HH\; \Box \delta \Box^{n-1}\GG \sqrt{-g} \; \dx - \int_M \nabla_\mu\HH \nabla_\nu \Box^{n-1}\GG \delta g^{\mu\nu} \sqrt{-g} \; \dx \\
-& \frac 12\int_M g_{\mu\nu}(\nabla^{\lambda} \HH \nabla_\lambda \Box^{n-1}\GG + \HH \Box \Box^{n-1}\GG) \delta g^{\mu\nu} \sqrt{-g} \; \dx \\
=&  \int_M \HH\; \Box\delta \Box^{n-1}\GG \sqrt{-g} \; \dx + \frac 12 \int_M S_{\mu\nu}(\HH,\Box^{n-1}\GG) \delta g^{\mu\nu} \sqrt{-g} \; \dx.
\end{aligned} \end{equation}

Partial integration in the first term of the previous formula yields
\begin{equation}\begin{aligned}
I=&  \int_M \Box \HH\; \delta \Box^{n-1}\GG \sqrt{-g} \; \dx + \frac 12 \int_M S_{\mu\nu}(\HH,\Box^{n-1}\GG) \delta g^{\mu\nu} \sqrt{-g} \; \dx.
\end{aligned} \end{equation}
after $n-1$ more steps
\begin{equation}\begin{aligned}
I=&  \int_M \Big(\Box^n \HH\; \delta \GG  + \frac 12\sum_{l=0}^{n-1} S_{\mu\nu}(\Box^l \HH,\Box^{n-1-l}\GG) \delta g^{\mu\nu}\Big) \sqrt{-g} \; \dx.
\end{aligned} \end{equation}
\end{proof}

\begin{theorem} \label{thm:var}
  Let $\HH$ and $\GG$ be scalar functions of scalar curvature, then
  \begin{align}
    \int_M \HH \delta ( \sqrt{-g}) \; \dx &= -\frac 12 \int_M g_{\mu\nu} \HH \delta g^{\mu\nu} \sqrt{-g} \; \dx, \label{thm:eom.1}\\
    \int_M \HH \delta R \sqrt{-g} \; \dx &= \int_M \left( R_{\mu\nu} \HH - K_{\mu\nu} \HH\right) \delta g^{\mu\nu} \sqrt{-g} \; \dx, \label{thm:eom.2}\\
    \int_M \HH \delta (\FF(\Box)\GG) \sqrt{-g} \; \dx &= \int_M \left( R_{\mu\nu} - K_{\mu\nu}\right) \left(\GG' \FF(\Box) \HH \right)
    \delta g^{\mu\nu} \sqrt{-g} \; \dx \nonumber \\
    &+ \frac 12 \sum_{n=1}^\infty f_n \sum_{l=0}^{n-1}\int_M S_{\mu\nu} (\Box^l \HH, \Box^{n-1-l}\GG) \delta g^{\mu\nu} \sqrt{-g} \; \dx, \label{thm:eom.3}
\end{align}
where $S_{\mu\nu}(A,B) = g_{\mu\nu} \nabla^\alpha A \nabla_\alpha B + g_{\mu\nu} A \Box B -2 \nabla_\mu A \nabla_\nu B$.
\end{theorem}
\begin{proof}
Equation \eqref{thm:eom.1} is a consequence of \eqref{lem:deltag.2}. \\
From Lemmas \ref{lem:deltag} and  \ref{lem:hkdeltag} we get
\begin{equation}\begin{aligned}
\int_M \HH \delta R \sqrt{-g} \; \dx &= \int_M \left( R_{\mu\nu}\HH \delta g^{\mu\nu} - \HH K_{\mu\nu} \delta g^{\mu\nu}\right) \sqrt{-g} \; \dx \\
&= \int_M \left( R_{\mu\nu} \HH - K_{\mu\nu} \HH  \right)\delta g^{\mu\nu} \sqrt{-g} \; \dx.
\end{aligned}\end{equation}
To prove \eqref{thm:eom.3} let us introduce the following notation
\begin{align}
J_n = \int_M \HH \delta(\Box^n \GG) \sqrt{-g} \; \dx.
\end{align}
Then,
\begin{align}
  \int_M \HH \delta (\FF(\Box)\GG) \sqrt{-g} \; \dx &= \sum_{n=0}^\infty f_n J_n.
\end{align}
The integral $J_0$ is calculated by applying \eqref{thm:eom.2}, i.e.
\begin{align}
  J_0 &= \int_M \left(R_{\mu\nu} \GG' \HH -K_{\mu\nu} \GG'\HH\right)\delta g^{\mu\nu} \sqrt{-g} \; \dx.
\end{align}
For $n>0$ integral $J_n$ is calculated by applying Lemma \ref{lem:j_n},
\begin{align}
J_n &= \int_M \Box^n \HH \delta\GG \sqrt{-g} \; \dx + \frac 12 \sum_{l=0}^{n-1}\int_M S_{\mu\nu}(\Box^l \HH,\Box^{n-1-l}\GG) \delta g^{\mu\nu} \; \sqrt{-g} \; \dx.
\end{align}

Using \eqref{thm:eom.2} in the first term we obtain
\begin{align}
J_n &= \int_M \left(R_{\mu\nu}\GG'\Box^n \HH -K_{\mu\nu}(\GG'\Box^n \HH)\right) \delta g^{\mu\nu} \sqrt{-g} \; \dx \nonumber \\
&+ \frac 12 \sum_{l=0}^{n-1}\int_M S_{\mu\nu}(\Box^l \HH,\Box^{n-1-l}\GG) \delta g^{\mu\nu} \; \sqrt{-g} \;
\dx.
\end{align}
Summation over $n$ yields the final result
\begin{align}
I&=  \int_M \HH \delta (\FF(\Box)\GG) \sqrt{-g} \; \dx = \sum_{n=0}^\infty f_n J_n \nonumber\\
&= \sum_{n=0}^\infty f_n\int_M \left(R_{\mu\nu}\GG'\Box^n \HH -K_{\mu\nu}(\GG'\Box^n \HH)\right) \delta g^{\mu\nu} \sqrt{-g} \; \dx \nonumber \\
&+ \frac 12 \sum_{n=1}^\infty\sum_{l=0}^{n-1}f_n\int_M S_{\mu\nu}(\Box^l \HH,\Box^{n-1-l}\GG) \delta g^{\mu\nu} \; \sqrt{-g} \;
\dx \\
  &= \int_M \left(R_{\mu\nu}\GG'\FF(\Box) \HH -K_{\mu\nu}(\GG'\FF(\Box) \HH)\right) \delta g^{\mu\nu} \sqrt{-g} \; \dx \nonumber\\
&+ \frac 12 \sum_{n=1}^\infty\sum_{l=0}^{n-1}f_n\int_M S_{\mu\nu}(\Box^l \HH,\Box^{n-1-l}\GG) \delta g^{\mu\nu} \; \sqrt{-g} \;
\dx.\nonumber  
\end{align}
\end{proof}

\section{Equations of motion}
Let us return to  action \eqref{action}. In order to calculate $\delta S$ we introduce the following auxiliary actions
\begin{align}
  S_0 &= \int_M (R-2\Lambda) \; \sqrt{-g} \; \dx, \\
  S_1 &= \int_M \HH(R) \FF(\Box) \GG(R) \; \sqrt{-g} \; \dx.
\end{align}
Action $S_0$ is Einstein-Hilbert action and its variation is
\begin{align}\label{eq:vareh}
  \delta S_0 &= \int_M \left(G_{\mu\nu} + \Lambda g_{\mu\nu}\right)\delta g^{\mu\nu}\; \sqrt{-g} \; \dx.
\end{align}

\begin{lemma} \label{lem:vars1}
Variation of the action $S_1$ is
\begin{align}
\delta S_1 &= -\frac{1}{2} \int_M g_{\mu\nu} \HH(R)\FF(\Box)\GG(R) \delta g^{\mu\nu}\; \sqrt{-g} \; \dx \nonumber\\
&+ \int_M \left(R_{\mu\nu} W - K_{\mu\nu} W \right) \delta g^{\mu\nu}\; \sqrt{-g} \; \dx \nonumber\\
&+ \frac{1}{2} \sum_{n=1}^{\infty} f_n \sum_{l=0}^{n-1} \int_M S_{\mu\nu}(\Box^l \HH(R), \Box^{n-1-l}\GG(R)) \delta g^{\mu\nu} \sqrt{-g} \; \dx,
\end{align}
where $W = \HH'(R) \FF(\Box)\GG(R) + \GG'(R)\FF(\Box)\HH(R)$. \phantom{A}
\end{lemma}
\begin{proof}
Variation $\delta S_1$ is equal to
\begin{align}
\delta S_1 &= \int_M \Big(\HH(R) \FF(\Box)\GG(R) \delta (\sqrt{-g}) + \delta \HH(R) \FF(\Box)\GG(R) \;\sqrt{-g} \nonumber \\
 &+  \HH(R) \delta(\FF(\Box)\GG(R)) \;\sqrt{-g}\Big) \dx.
\end{align}

All the terms in the previous formula are obtained by Theorem \ref{thm:var}. In particular \eqref{thm:eom.1} yields
\begin{align}\label{eq:2}
\int_M \HH(R) \FF(\Box)\GG(R) \delta (\sqrt{-g}) \; \dx &= -\frac 12 \int_M g_{\mu\nu} \HH(R)\FF(\Box)\GG(R) \; \delta g^{\mu\nu} \sqrt{-g} \; \dx.
\end{align}
Also, from equation \eqref{thm:eom.2} we get
\begin{align}
&\int_M \delta (\HH(R)) \FF(\Box)\GG(R) \; \sqrt{-g} \; \dx = \int_M \HH'(R) \delta R \; \FF(\Box)\GG(R) \;\sqrt{-g} \; \dx \nonumber \\
&=  \int_M  \Big( R_{\mu\nu}  \HH'(R)  \FF(\Box)\GG(R) -K_{\mu\nu} \left(\HH'(R) \FF(\Box)\GG(R) \right)\Big) \; \delta g^{\mu\nu} \sqrt{-g} \; \dx. \label{eq:3}
\end{align}
The last term is calculated by \eqref{thm:eom.3}.
\begin{align}
\int_M& \HH(R) \delta(\FF(\Box)\GG(R)) \;\sqrt{-g} \dx  \nonumber \\
&= \int_M \Big(R_{\mu\nu}  \GG'(R)  \FF(\Box)\HH(R) -K_{\mu\nu} \left(\GG'(R) \FF(\Box)\HH(R) \right) \Big) \; \delta g^{\mu\nu} \sqrt{-g} \; \dx \nonumber\\
&+ \sum_{n=1}^\infty f_n \sum_{l=0}^{n-1}\int_M S_{\mu\nu}\left(\Box^l \HH(R), \Box^{n-1-l}\GG(R) \right) \; \delta g^{\mu\nu} \sqrt{-g} \; \dx.\label{eq:4}
\end{align}
 Adding equations \eqref{eq:2}, \eqref{eq:3} and \eqref{eq:4} together proves the Lemma.
 \end{proof}
 \begin{theorem} \label{thm:eom}
Variation of the action \eqref{action} is equal to zero iff
\begin{align} \label{eom}
\hat G_{\mu\nu} &= G_{\mu\nu} +\Lambda g_{\mu\nu} -\frac{1}{2} g_{\mu\nu} \HH(R)\FF(\Box)\GG(R) + \left(R_{\mu\nu} W - K_{\mu\nu} W \right) +\frac 12 \Omega_{\mu\nu} = 0,
\end{align}
where
\begin{align}
  W &= \HH'(R) \FF(\Box)\GG(R) + \GG'(R) \FF(\Box)\HH(R), \\
  \Omega_{\mu\nu} &= \sum_{n=1}^{\infty} f_n \sum_{l=0}^{n-1} S_{\mu\nu}\left( \Box^l \HH(R), \Box^{n-1-l}\GG(R)\right).
\end{align}

\end{theorem}
\begin{proof}
The proof of  Theorem 2 is evident from the Lemma \ref{lem:vars1} and Theorem \ref{thm:var}.
\end{proof}

\section{Second variation of the action}

In this section we set $h_{\mu\nu} = \delta g_{\mu\nu}$. From Lemma \ref{lem:deltag} we see that $h^{\mu\nu} = - \delta g^{\mu\nu}$. Also let $h= g^{\mu\nu}h_{\mu\nu}$ be the trace of $h_{\mu\nu}$.

Operator $\delta \Box$ is defined by $(\delta \Box)V = \delta(\Box V) - \Box \delta V$. Then we can prove the following Lemma
\begin{lemma}
  Let $U,V$ be scalar functions. Then
  \begin{align}
    (\delta \Box)V &= -h^{\mu\nu} \nmu \nnu V - \nabla^\mu h_\mu^\lambda \nabla_\lambda V + \frac 12 \nabla^\lambda h \nabla_\lambda V, \\
    \int_M U (\delta \Box)V \sqrt{-g} \; \dx &= \frac 12 \int_M S_{\mu\nu}(U ,V) \delta g^{\mu\nu} \sqrt{-g} \; \dx.
  \end{align}
\end{lemma}
\begin{proof}
  For the first part, start with
  \begin{align}
    (\delta \Box)V &= \delta(\Box V) - \Box \delta V \\
    &= \delta(g^{\mu\nu} \nmu\nnu V) - \Box \delta V \\
    &=-h^{\mu\nu} \nmu\nnu V  -  g^{\mu\nu}\delta \Gamma_{\mu\nu}^\lambda \nabla_\lambda V \\
    &=-h^{\mu\nu} \nmu\nnu V  -  \frac 12 g^{\mu\nu}(\nmu h_\nu^\lambda + \nnu h_\mu^\lambda - \nabla^\lambda h_{\mu\nu}) \nabla_\lambda V \\
    &= -h^{\mu\nu} \nmu \nnu V - \nabla^\mu h_\mu^\lambda \nabla_\lambda V + \frac 12 \nabla^\lambda h \nabla_\lambda V.
  \end{align}
  The second part of the Lemma is proved by

  \begin{align}
    &\int_M U (\delta \Box)V \sqrt{-g} \; \dx \\
    &= \int_M U (-h^{\mu\nu} \nmu \nnu V - \nabla^\mu h_\mu^\lambda \nabla_\lambda V + \frac 12 \nabla^\lambda h \nabla_\lambda V) \sqrt{-g} \; \dx \\
    &= -\int_M (U  \nmu \nnu V -\nabla_\mu( U \nabla_\nu V) + \frac 12 g_{\mu\nu} \nabla^\lambda(U\nabla_\lambda V)) h^{\mu\nu} \sqrt{-g} \; \dx \\
    &= \frac 12 \int_M (-2\nabla_\mu U \nabla_\nu V  + g_{\mu\nu} \nabla^\lambda U \nabla_\lambda V + g_{\mu\nu} U\Box V) \delta g^{\mu\nu} \sqrt{-g} \; \dx \\
    &= \frac 12 \int_M S_{\mu\nu}(U,V) \delta g^{\mu\nu} \sqrt{-g} \; \dx.
  \end{align}
\end{proof}
In the next lemma we find the variation of $\FF(\Box)$.

\begin{lemma} \label{lem:deltafbox}
Let $U, V$ be scalar functions. Then,
  \begin{align}
    \int_M U \delta(\FF(\Box))V \sqrt{-g} \; \dx = \sum_{n=1}^{\infty} f_n \sum_{l=0}^{n-1} \int_M S_{\mu\nu}(\Box^l U, \Box^{n-1-l} V) \delta g^{\mu\nu} \sqrt{-g} \; \dx .
  \end{align}
\end{lemma}
\begin{proof}
Note that $\delta\Box^n = \sum_{l=0}^{n-1} \Box^l (\delta\Box) \Box^{n-1-l}$ for $n>0$ and $\delta \Box^0 = \delta \mathrm{Id} = 0$. Therefore summation  over $n$ and integration yields
  \begin{align}
    &\int_M U \delta(\FF(\Box))V \sqrt{-g} \; \dx = \sum_{n=1}^{\infty} f_n \sum_{l=0}^{n-1} \int_M U \Box^l (\delta\Box) \Box^{n-1-l} V \sqrt{-g} \; \dx \\
    &= \sum_{n=1}^{\infty} f_n \sum_{l=0}^{n-1} \int_M \Box^l U (\delta\Box) \Box^{n-1-l} V \sqrt{-g} \; \dx \\
    &= \sum_{n=1}^{\infty} f_n \sum_{l=0}^{n-1} \int_M S_{\mu\nu}(\Box^l U,  \Box^{n-1-l} V ) \delta g^{\mu\nu} \sqrt{-g} \; \dx.
  \end{align}
\end{proof}

\begin{lemma} \label{lem:deltaw}
Let $U$ be scalar function. Then,
\begin{align}
  &\int_M U \delta W  \sqrt{-g} \; \dx = \int_M (R_{\mu\nu} Y - K_{\mu\nu} Y + \frac 12 \Psi_{\mu\nu}) \delta g^{\mu\nu} \sqrt{-g} \; \dx, \\
  Y &= U(P'' \FF(\Box)Q +Q''\FF(\Box)P) + (P' \FF(\Box)(Q' U) + Q' \FF(\Box)(P' U)), \\
  \Psi_{\mu\nu} &= \sum_{n=1}^{+\infty} f_n \sum_{l=0}^{n-1} \left( S_{\mu\nu}(\Box^l (P'U), \Box^{n-1-l} Q) + S_{\mu\nu}(\Box^l (Q'U), \Box^{n-1-l} P) \right) .
\end{align}
\end{lemma}
\begin{proof}
Since $W = P' \FF(\Box) Q + Q' \FF(\Box) P$ the variation $\delta W$ is written as
\begin{align}
\delta W &= P'' \FF(\Box) Q \delta R + P' \delta(\FF(\Box)) Q + P' \FF(\Box) (Q' \delta R) \nonumber \\
&+ Q'' \FF(\Box) P \delta R + Q' \delta(\FF(\Box)) P + Q' \FF(\Box) (P' \delta R).
\end{align}
Integration of the second  and fifth term in this sum is done by using Lemma \ref{lem:deltafbox}. The remaining four terms are obtained by Theorem \ref{thm:var}.
\end{proof}

\begin{lemma} \label{lem:deltasmunu}
Let $A,B$ be scalar functions. Then,
\begin{align}
  \int_M S_{\mu\nu}(\delta A, B) \delta g^{\mu\nu} \sqrt{-g}\; \dx &= \int \sigma_1(B) \delta A\sqrt{-g}\; \dx, \\
  \int_M S_{\mu\nu}(A, \delta B) \delta g^{\mu\nu} \sqrt{-g}\; \dx &= \int \sigma_2(A) \delta B\sqrt{-g}\; \dx,
\end{align}
where
\begin{align}
  \sigma_1(B) &= \nabla^\lambda h \nabla_\lambda B - 2\nabla_\mu h^{\mu\nu} \nabla_\nu B -2 h^{\mu\nu} \nmu\nnu B, \\
  \sigma_2(A) &= - \nabla^\lambda h \nabla_\lambda A - A\Box h - 2\nnu h^{\mu\nu} \nmu A -2 h^{\mu\nu} \nmu\nnu A.
\end{align}
\end{lemma}
\begin{proof}
To prove the first equation recall the definition of $S_{\mu\nu}(A,B)$
\begin{align}
  &\int_M S_{\mu\nu}(\delta A, B) \delta g^{\mu\nu} \sqrt{-g}\; \dx \\
  &= \int_M \left( g_{\mu\nu} \nabla^\alpha \delta A \nabla_\alpha B + g_{\mu\nu} \delta A \Box B -2 \nabla_\mu \delta A \nabla_\nu B \right) \delta g^{\mu\nu}\sqrt{-g}\; \dx \\
  &= \int_M \left(- h \nabla^\alpha \delta A \nabla_\alpha B - h \delta A \Box B +2 h^{\mu\nu} \nabla_\mu \delta A \nabla_\nu B \right) \sqrt{-g}\; \dx \\
  &= \int_M \left(  \nabla^\alpha(h  \nabla_\alpha B) - h  \Box B - 2  \nabla_\mu(h^{\mu\nu} \nabla_\nu B) \right)\delta A \sqrt{-g}\; \dx  \\
  &= \int_M \left(  \nabla^\alpha h  \nabla_\alpha B - 2  \nabla_\mu h^{\mu\nu} \nabla_\nu B -2 h^{\mu\nu} \nmu \nnu B \right)\delta A \sqrt{-g}\; \dx \\
  &= \int_M \sigma_1(B)\delta A \sqrt{-g}\; \dx.
\end{align}

The proof of the second equation is similar.
\end{proof}

\begin{lemma} \label{lem:deltaomega}
Let $\Omega_{\mu\nu} = \sum_{n=1}^{\infty} f_n \sum_{l=0}^{n-1} S_{\mu\nu}\left( \Box^l \HH(R), \Box^{n-1-l}\GG(R)\right)$. Then,
\begin{align}
  &\int_M \delta \Omega_{\mu\nu} \delta g^{\mu\nu} \sqrt{-g}\; \dx \nonumber \\
  &= \int_M \sum_{n=1}^{+\infty} f_n \sum_{l=0}^{n-1} \Big(h_{\mu\nu} \nabla^\lambda \Box^l P \nabla_\lambda \Box^{n-1-l} Q + h \nabla_\mu \Box^l P \nabla_\nu \Box^{n-1-l} Q \nonumber \\
  &+ h_{\mu\nu} \Box^l P \Box^{n-l} Q - \frac 12 S_{\mu\nu}(h \Box^l P, \Box^{n-1-l} Q)  \nonumber \\
  &+ R_{\mu\nu} P' \Box^l (\sigma_1(\Box^{n-1-l} Q)) -  K_{\mu\nu} (P' \Box^l (\sigma_1(\Box^{n-1-l} Q))) \nonumber \\
  &+ R_{\mu\nu} Q' \Box^l (\sigma_2(\Box^{n-1-l} P)) -  K_{\mu\nu} (Q' \Box^l (\sigma_2(\Box^{n-1-l} P)))\Big)\delta g^{\mu\nu} \sqrt{-g}\; \dx \nonumber\\
  &+ \frac 12 \int_M \sum_{n=1}^{+\infty} f_n \sum_{l=1}^{n-1} \sum_{m=0}^{l-1}\Big( S_{\mu\nu}(\Box^m(\sigma_1(\Box^{n-1-l} Q)),\Box^{l-m-1}P)\nonumber \\
  &+ S_{\mu\nu}(\Box^m(\sigma_2(\Box^{n-1-l} P)),\Box^{l-m-1}Q)\Big)\delta g^{\mu\nu} \sqrt{-g}\; \dx .
\end{align}
\end{lemma}
\begin{proof}
Note that
\begin{align}
  \delta \Omega_{\mu\nu} &= \sum_{n=1}^{\infty} f_n \sum_{l=0}^{n-1} \delta S_{\mu\nu}\left( \Box^l \HH(R), \Box^{n-1-l}\GG(R)\right).
\end{align}
Moreover,
\begin{align}
  &\int_M \delta S_{\mu\nu}(A,B) \delta g^{\mu\nu} \sqrt{-g}\; \dx = \int_M \left(S_{\mu\nu}(\delta A,B) + S_{\mu\nu}(A,\delta B)\right) \delta g^{\mu\nu}\sqrt{-g}\; \dx\nonumber \\
  &+ \int_M \big(h_{\mu\nu} \nabla^\lambda A \nabla_\lambda B + h \nabla_\mu A \nabla_\nu B \nonumber \\
  &+ h_{\mu\nu} A \Box B - \frac 12 S_{\mu\nu}(h A, B)\big) \delta g^{\mu\nu}\sqrt{-g}\; \dx.
\end{align}
Using this formula for each term in $\delta \Omega_{\mu\nu}$ yields the result of the Lemma.
\end{proof}

\begin{theorem} \label{thm:secondvar}
  The second variation of the action \eqref{action} is given by
\begin{align}
  \delta^2 S &= \frac 1{16\pi G} \int_M \left( U_{\mu\nu} + R_{\mu\nu} X - K_{\mu\nu} X + \frac 12 \chi_{\mu\nu} + \frac 14 \Theta_{\mu\nu}\right) \delta g^{\mu\nu} \sqrt{-g}\; \dx,
\end{align}
where
\begin{align}
U_{\mu\nu} &= - \frac 12 h_{\mu\nu} (R-2\Lambda + P\FF(\Box) Q) + \delta R_{\mu\nu} (W+1) + \delta \Gamma_{\mu\nu}^\lambda \nabla_\lambda W \nonumber \\
&+ h_{\mu\nu} \Box W - \frac 12 S_{\mu\nu}(h,W),\\
X &= \frac 12 (h + P' h \FF(\Box) Q + Q' \FF(\Box)(Ph)) + \Big( \delta R(P'' \FF(\Box)Q +Q''\FF(\Box)P) \nonumber \\
&+ (P' \FF(\Box)(Q' \delta R) + Q' \FF(\Box)(P' \delta R )) \Big), \\
\chi_{\mu\nu} &= \frac 12 \sum_{n=1}^{+\infty}f_n \sum_{l=0}^{n-1}  S_{\mu\nu}(\Box^l (P h), \Box^{n-l-1} Q) \nonumber \\
  &- \sum_{n=1}^{+\infty}f_n \sum_{l=0}^{n-1} \left( S_{\mu\nu}(\Box^l (P'M), \Box^{n-1-l} Q) + S_{\mu\nu}(\Box^l (Q'M), \Box^{n-1-l} P) \right) \nonumber \\
  &+ \sum_{n=1}^{+\infty}f_n \sum_{l=0}^{n-1}  \Big(h_{\mu\nu} \nabla^\lambda \Box^l P \nabla_\lambda \Box^{n-1-l} Q + h \nabla_\mu \Box^l P \nabla_\nu \Box^{n-1-l} Q \nonumber \\
  &+  h_{\mu\nu} \Box^l P \Box^{n-l} Q - \frac 12 S_{\mu\nu}(h \Box^l P, \Box^{n-1-l} Q)  \nonumber \\
  &+ (R_{\mu\nu}-K_{\mu\nu}) (P' \Box^l (\sigma_1(\Box^{n-1-l} Q)) + Q' \Box^l (\sigma_2(\Box^{n-1-l} P))) \Big),
\end{align}
and
\begin{align}
\Theta_{\mu\nu} &= \sum_{n=1}^{+\infty} f_n \sum_{l=1}^{n-1} \sum_{m=0}^{l-1}\Big( S_{\mu\nu}(\Box^m(\sigma_1(\Box^{n-1-l} Q)),\Box^{l-m-1}P) \nonumber \\
  &+ S_{\mu\nu}(\Box^m(\sigma_2(\Box^{n-1-l} P)),\Box^{l-m-1}Q)\Big),\\
  \sigma_1(B) &= \nabla^\lambda h \nabla_\lambda B - 2\nabla_\mu h^{\mu\nu} \nabla_\nu B -2 h^{\mu\nu} \nmu\nnu B, \\
  \sigma_2(A) &= - \nabla^\lambda h \nabla_\lambda A - A\Box h - 2\nnu h^{\mu\nu} \nmu A -2 h^{\mu\nu} \nmu\nnu A.
\end{align}

\end{theorem}

\begin{proof}
In the pervious section we calculated the first variation of the action \eqref{action}
\begin{align}
  \delta S &= \frac 1{16\pi G} \int_M \hat G_{\mu\nu} \delta g^{\mu\nu} \sqrt{-g} \; \dx.
\end{align}
Moreover the second variation $\delta^2 S$ is
\begin{align}
    \delta^2 S &= \frac 1{16\pi G} \int_M \left(\delta\hat G_{\mu\nu} \delta g^{\mu\nu} + \hat G_{\mu\nu} \delta^2 g^{\mu\nu} - \frac 12 g_{\alpha\beta} \delta g^{\alpha\beta} \hat G_{\mu\nu} \delta g^{\mu\nu}\right)  \sqrt{-g} \; \dx.
\end{align}

At the beginning note that
\begin{align}
  &\int_M \delta \left( G_{\mu\nu} +\Lambda g_{\mu\nu}\right) \delta g^{\mu\nu} \sqrt{-g} \; \dx\\
  &= \int_M \left( \delta R_{\mu\nu} - \frac 12(R-2\Lambda) h_{\mu\nu} + \frac 12R_{\mu\nu} h - \frac 12 K_{\mu\nu} h \right) \delta g^{\mu\nu} \sqrt{-g} \; \dx.
\end{align}

The next term is calculated by using Lemma \ref{lem:deltafbox}
\begin{align}
  &\int_M \delta\left( g_{\mu\nu} P\FF(\Box) Q\right) \delta g^{\mu\nu} \sqrt{-g}\; \dx \\
  =& \int_M h_{\mu\nu} P\FF(\Box) Q \delta g^{\mu\nu} \sqrt{-g}\; \dx + \int g_{\mu\nu} P \delta(\FF(\Box))Q  \delta g^{\mu\nu} \sqrt{-g}\; \dx \nonumber\\
  -& \int_M \left(P' h \FF(\Box) Q + Q' \FF(\Box)(P h \right)) \delta R \sqrt{-g}\; \dx \\
  =& \int_M h_{\mu\nu} P\FF(\Box) Q \delta g^{\mu\nu} \sqrt{-g}\; \dx \nonumber \\
  -& \frac 12  \sum_{n=1}^{+\infty}f_n \sum_{l=0}^{n-1} \int S_{\mu\nu}(\Box^l (P h), \Box^{n-l-1} Q) \delta g^{\mu\nu}  \sqrt{-g}\; \dx \nonumber\\
  -& \int_M R_{\mu\nu} \left(P' h \FF(\Box) Q + Q' \FF(\Box)(P h \right)) \delta g^{\mu\nu} \sqrt{-g}\; \dx \nonumber \\+& \int K_{\mu\nu} \left(P' h \FF(\Box) Q + Q' \FF(\Box)(P h \right)) \delta g^{\mu\nu} \sqrt{-g}\; \dx.
\end{align}

The third term is $\int_M \delta\left(R_{\mu\nu} W\right) \delta g^{\mu\nu} \sqrt{-g}\; \dx$ and it is equal to
\begin{align}
  &\int_M \delta\left(R_{\mu\nu} W\right) \delta g^{\mu\nu} \sqrt{-g}\; \dx \\
  =&   \int_M \left(W \delta R_{\mu\nu} + R_{\mu\nu} \delta W \right)\delta g^{\mu\nu} \sqrt{-g}\; \dx \\
  =& -\frac 12 \int_M \left(\Box h_{\mu\nu} +\nabla_\mu \nabla_\nu h - 2 \nabla_\lambda \nabla_\mu h_\nu^\lambda\right) W\delta g^{\mu\nu} \sqrt{-g}\; \dx \nonumber \\
  &+ \int_M  R_{\mu\nu} \delta W \delta g^{\mu\nu} \sqrt{-g}\; \dx .
\end{align}

The last integral of the above formula is obtained by Lemma \ref{lem:deltaw}. Similarly, we obtain
\begin{align}
  &\int_M \delta\left(K_{\mu\nu} W\right) \delta g^{\mu\nu} \sqrt{-g}\; \dx \\
  =& \int_M K_{\mu\nu} \delta W \delta g^{\mu\nu} \sqrt{-g}\; \dx \nonumber\\
  &+ \int \left(\delta \Gamma_{\mu\nu}^\lambda \nabla_\lambda W + h_{\mu\nu} \Box W + g_{\mu\nu} (\delta \Box) W \right) h^{\mu\nu} \sqrt{-g} \; \dx \\
  =& \int_M  \delta W K_{\mu\nu} \delta g^{\mu\nu} \sqrt{-g}\; \dx \nonumber \\
  &- \int \left(\delta \Gamma_{\mu\nu}^\lambda \nabla_\lambda W + h_{\mu\nu} \Box W - \frac 12 S_{\mu\nu} (h,W) \right) \delta g^{\mu\nu} \sqrt{-g} \; \dx .
\end{align}

At the end the last term $\int_M \delta \Omega_{\mu\nu}  \delta g^{\mu\nu} \sqrt{-g}\; \dx$ is calculated in Lemma \ref{lem:deltaomega}.

\end{proof}

\section{Perturbations}

\subsection{Background}

In this section we start with $FRW$ metric, which for $k=0$  can be written as
\begin{equation}
ds^2=-dt^2+a(t)^2(\mathrm d x^2 + \mathrm d y^2 + \mathrm d z^2) .
\label{FRW}
\end{equation}
Some relevant background quantities  are
\begin{equation}
R=12H^2+6\dot H ,~\Gamma_{ij}^0=Hg_{ij},~\Gamma_{j0}^i=H\delta^i_{j},~
 \Box=-\partial_t^2-3H\partial_t .
\end{equation}

For perturbations it is useful to employ the canonical ADM $(1+3)$ decomposition and
introduce the conformal time $\tau$ such that
$$
a\mathrm d\tau = dt .
$$
Then the flat FRW metric \eqref{FRW} transforms to
\begin{equation}
ds^2=a(\tau)^2(-d\tau^2+\mathrm d x^2 + \mathrm d y^2 + \mathrm d z^2).
\label{dScosmoconf}
\end{equation}

\subsection{Perturbations}

The metric  perturbations (see \cite{muhanov}) can be divided into three types: scalar, vector and tensor perturbations.
The component $h_{00}$ is invariant under spatial rotations and translations and therefore
\begin{equation}
  h_{00} = 2a(\tau)^2 \phi.
\end{equation}
The components $h_{0i}$ are the sum of a spatial gradient of a function $B$ and divergence free vector $S_i$.
\begin{equation}
  h_{0i} = a(\tau)^2 (\partial_i B +S_i).
\end{equation}
Similarly, components $h_{ij}$, which transform as a tensor under $3$-rotations are written as
\begin{equation}
  h_{ij} = a(\tau)^2 (2\psi \delta_{ij} + 2 \partial_{ij}^2 E + \partial_i F_j +\partial_j F_i + \varphi_{ij}),
\end{equation}
where $\psi$ and $E$ are scalar functions, $F_i$ is a vector with zero divergence and $3$-tensor satisfies
\begin{equation}
  \varphi_i^i = 0, \qquad \partial_i \varphi_j^i = 0.
\end{equation}

Note that there are four scalar functions, two vectors with two independent components each and tensor $\varphi_{ij}$ has two independent components. Therefore, as expected we have in total ten functions.

The scalar perturbations are defined  by scalar functions $\phi, \psi, B, E$ and perturbed metric around $FRW$  background is
\begin{equation}
ds^2=a(\tau)^2\left[-(1-2\phi)d\tau^2+\partial_i B \mathrm d\tau
dx^i+((1+2\psi)\delta_{ij}+2\partial_i\partial_j E)dx^idx^j\right].
\label{scalar_pert}
\end{equation}
The vector perturbations are defined by vectors $S_i$ and $F_i$, i.e.
\begin{equation}
ds^2=a(\tau)^2\left[- d\tau^2+S_i B \mathrm d\tau
dx^i+(\delta_{ij}+\partial_iF_j +\partial_j F_i)dx^idx^j\right].
\label{vector_pert}
\end{equation}
Tensor perturbations  are defined by $\varphi_{ij}$ and describe gravitational  waves, which have no analog in Newton gravity theory.
\begin{equation}
ds^2=a(\tau)^2\left[- d\tau^2+(\delta_{ij}+\varphi_{ij})dx^idx^j\right].
\label{tensor pert}
\end{equation}
Each of the types of perturbations can be studied separately.
In this form of perturbations we have
\begin{align}
  h_{\mu\nu} &= a(\tau)^2 \left(
                           \begin{array}{cc}
                             2\phi & \partial_i B +S^i \\
                             \partial_i B +S^i & 2\psi \mathrm{Id} + 2\Hess E +\partial_j F_i +\partial_i F_j +\varphi_{ij} \\
                           \end{array}
                         \right), \\
  h &= -2\phi+6\psi +2 \triangle E, \\
  \delta R &= \frac 2{a^2} \big(6 \frac{a''}a \phi + \triangle(\phi-2\psi) + 3\frac{a'}a \triangle(B+E') \nonumber \\
  &+ 3\frac{a'}a (\phi' + 3\psi')+\triangle(B'+E'') + 3\psi'' \big).
\end{align}
Moreover, let $A^\mu = \nabla_\nu h^{\mu\nu}$, then
\begin{align}
A^0 &= 6\frac{a'}{a^3}(\phi+\psi) + \frac 1{a^2}\triangle B + 2\frac{a'}{a^3}\triangle E + \frac{2}{a^2}\phi', \\
A^i &= 4 \frac{a'}{a^3}(\partial_i B +S^i) + \frac 2{a^2}(\partial_i \psi + \triangle\partial_i E) + \frac{1}{a^2}(\partial_i B' +S'^i+\triangle F^i).
\end{align}

Out of 4 scalar modes only 2 are gauge invariant.
The convenient gauge invariant variables (Bardeen potentials) are introduced as
\begin{equation}
\Phi=\phi-\frac{1}{a}(a(B-E^\prime))^\prime,\qquad\Psi=\psi+\frac {a'}a(B-E^\prime).
\label{GIvars}
\end{equation}
The prime denotes the differentiation with respect to the conformal time
$\tau$ and the dot as before w.r.t. the cosmic time $t$.

The $(1+3)$ structure suggests
to represent the perturbation quantities (which can depend on all 4
coordinates) as
\begin{equation}
f(\tau, \vec x)=f(\tau,k)Y(k,\vec x),
\label{scalar13}
\end{equation}
where $\vec x = (x,y,z)$ and  $k=|\vec k|$ comes from the definition of the $Y$-functions as spatial
Fourier modes
\begin{equation}
\delta^{ij}\partial_i\partial_j Y=-k^2Y .
\end{equation}
Then
\begin{equation}
Y=Y_0e^{\pm i\vec k\vec x} .
\end{equation}
The relevant expressions for the d'Alembert operator are
\begin{equation}
\Box=-\frac{1}{a^2}\partial_\tau^2-2\frac{a'}{a^3}\partial_\tau +\frac{\delta^{ij}\partial_i\partial_j}{a^2}
=-\partial_t^2-3H\partial_t-\frac{k^2}{a^2} ,
\end{equation}
where $k=|\vec k|$.

All the expression in this subsection are valid for a generic scale
factor $a$ in flat space.

\section{Concluding Remarks}

In this paper we have considered a class of nonlocal gravity models without
matter given by the action in the form
\begin{align} \label{action*}
  S &= \frac{1}{16\pi G} \int_M \left(R-2\Lambda + \HH(R) \FF(\Box) \GG(R) \right)\; \sqrt{-g} \; \dx.
\end{align}
We have derived the equations of motion for this action.
We also have presented the second variation of   action \eqref{action*} and basics of metric perturbations.

In many research papers there are equations of motion which are
special cases of our equations. In the case
$\HH(R)=\GG(R)= R$ one obtains
\begin{align*}
  S &= \frac{1}{16\pi G} \int_M \left(R-2\Lambda + R \FF(\Box)R \right)\; \sqrt{-g} \; \dx.
\end{align*}
This nonlocal model is further elaborated in the series of papers \cite{biswas1,biswas2,biswas3,biswas4,biswas5,koshelev,dimitrijevic1,dimitrijevic2,dimitrijevic3,dimitrijevic7,dimitrijevic8}.

The action \eqref{action*} for $\HH(R)= R^{-1}$ and $\GG(R)= R$ was introduced in \cite{dimitrijevic3} as a new approach to nonlocal gravity. This model one can also find in \cite{dimitrijevic4}.

The case $\HH(R)= R^p $ and $\GG(R)= R^q$ we analyzed in \cite{dimitrijevic7,dimitrijevic8}.

For the case $\HH(R)= (R+R_0)^m $ and $\GG(R)= (R+R_0)^m$ see \cite{dimitrijevic5,dimitrijevic6}.

 Studies of this model with $R = \mathrm{const}$ can be found in \cite{dimitrijevic9,dimitrijevic10}.

It is worth noting that cosmology with nonlocality in the matter sector was also investigated, see e.g. \cite{arefeva,eliz}.

For some very recent achievements in higher derivative modified gravities one can see \cite{koshelev-1,koshelev-2,koshelev-3,koshelev-4}.

Note that  there is the following formula $ \Box^{-1} = \int_0^\infty e^{-\alpha \Box} \, d\alpha ,$ which could be  used in investigation of models containing $\Box^{-n}, \, \,  n \in \mathbb{N} ,$
where $$\Box^{-n}  = \frac{1}{(n-1)!} \int_0^\infty  \alpha^{n-1}  e^{-\alpha \Box} \, d\alpha \, .$$ Namely, formalism presented in the previous sections can easily incorporate this case taking $\mathcal{F}(\Box) = e^{-\alpha \Box}$ and at the end performing integration over $\alpha$.

\begin{acknowledgement}
Work on this paper was partially supported by the Ministry of Education, Science and Technological Development of  Republic of Serbia, grant No 174012.
B.D. thanks Prof. Vladimir Dobrev for invitation to participate and give a talk on nonlocal gravity, as well as for hospitality, at the X International Symposium
``Quantum Theory and Symmetries'', and
XII International Workshop ``Lie Theory
and its Applications in Physics'', 19--25 June 2017, Varna, Bulgaria. B.D. also thanks a support of the ICTP - SEENET-MTP project NT-03 Cosmology-Classical and Quantum Challenges during preparation of this article.
\end{acknowledgement}

\end{document}